\documentclass[11pt]{article}

\usepackage{amsmath}
\usepackage{amsfonts}
\usepackage{amssymb}
\usepackage{graphicx}
\usepackage{caption}
\usepackage{subcaption}
\usepackage{color}
\usepackage{graphics}
\usepackage{epstopdf}

\usepackage{times}
\newtheorem{theorem}{Theorem}

\newtheorem{lemma}{Lemma}

\newtheorem{definition}{Definition}

\numberwithin{equation}{section}
\newenvironment{proof}[1][Proof]{\noindent\textbf{#1.} }{\ \rule{0.5em}{0.5em}}

\renewcommand{\epsilon}{\varepsilon}

\def\M{\mathcal{M}}
\def\T{\mathcal{T}}

\def\D{\mathcal{D}}

\def\bE{\mathbb{E}}
\def\Pr{\mathrm{Pr}}
\makeatletter
\let\@copyrightspace\relax
\makeatother

\begin{document}

\title{Quantum Advantage on Information Leakage for Equality}

\author{Juan Miguel Arrazola\thanks{Centre for Quantum Technologies, National University of Singapore, and Institute for Quantum Computing, University of Waterloo.
}
\and Dave Touchette \thanks{Institute for Quantum Computing and Department of Combinatorics and Optimization,
University of Waterloo, and the Perimeter Institute for Theoretical Physics.
}}

\maketitle

\vspace{2cm}

\begin{abstract}
We  prove a  lower bound on the information leakage of any classical protocol computing the equality function in the simultaneous message passing (SMP) model. Our bound is valid in the finite length regime and is strong enough to demonstrate a quantum advantage in terms of information leakage for practical quantum protocols. We prove our bound by obtaining an improved finite size version of the communication bound due to Babai and Kimmel, relating randomized communication to deterministic communication in the SMP model. We then relate information leakage to randomized communication through a series of reductions. We first provide alternative characterizations for information leakage, allowing us to link it to average length communication while allowing for shared randomness (pairwise, with the referee). A Markov inequality links this with bounded length communication, and a Newman type argument allows us to go from shared to private randomness. The only reduction in which we incur more than a logarithmic additive factor is in the Markov inequality; in particular, our compression method is essentially tight for the SMP model with average length communication.
\end{abstract}

% SSSSSSSSSSSSSSSSSSSSSSSSSSSSSSSSSSSSSSSSSSSSSSSSSS %
\section{Introduction}
% SSSSSSSSSSSSSSSSSSSSSSSSSSSSSSSSSSSSSSSSSSSSSSSSSS %

The simultaneous message passing model (SMP) of communication complexity (CC) can be understood as a 
simple model of a three node network with two players, Alice and Bob, and a referee, Charlie. 
Alice is given some input $x$, Bob some input $y$, and each sends a message to Charlie who should 
be able to compute with high probability some function $f(x,y)$ from the messages. 
Alice and Bob are not allowed to pre-share any resource, e.g.~randomness or entanglement.
See Ref.~\cite{KN97} for an introduction to classical communication complexity.
Here, we are interested in $f$ being the equality function, i.e. Charlie must determine whether $x = y$ or not. 

It is known that quantum protocols have 
an exponential advantage in terms of communication for computing equality in this model~\cite{BCWW01}. Communication is always an upper bound on the information leakage of a protocol, so  the best 
theoretical quantum protocol has information leakage logarithmic in the input size. 
In contrast, information leakage of any classical protocol computing equality is known to be at 
least quadratic in the input size~\cite{CSWY01}. Hence, a three node quantum 
network enables Alice and Bob, by each sending a single message to Charlie, to let him verify whether or not their inputs agree, while revealing exponentially less information to Charlie about these inputs than would be possible in any classical network. 

In the context of quantum communication, the task of computing equality in the SMP model is referred to as quantum fingerprinting. Recently, Arrazola and L\"utkenhaus~\cite{AL14} have proposed a practical quantum fingerprinting protocol which has prompted experimental efforts in this field~\cite{XAWW15,GX16}. However, the bounds stated above are asymptotic and do not account for any effect due to finite size inputs that would be relevant in a practical setting. 
Moreover, although finite-length lower bounds for the communication cost of classical protocols are known which can be beaten by quantum protocols, similar lower bounds on the information leakage are asymptotic in nature and the hidden constants are too large for quantum protocols to surpass the bounds in a practical setting.

In order to show that quantum protocols can achieve a smaller information leakage than any 
classical protocol for some finite length inputs, we must improve on previously known lower bounds and account for the finite size nature of practical protocols. In this work, we start by improving on the known communication bounds and get an improvement of more than one order of magnitude. We then leverage this communication bound to an information leakage lower bound by using a series of reductions and by providing alternative characterizations of information leakage 
in the SMP model. We show that, up to a logarithmically small additive term, the average cost of communication 
exactly agrees with the information leakage; this is the best statement we can hope for.
The main technical ingredient we use is a single message compression result due to Harsha, Jain, McAllester, and Radhakrishnan~\cite{HJMR10}.
Markov's inequality and a Newman type argument then allows us to link information leakage of any classical protocol to our improved lower bound on worst-case communication. This new bound can be used to show that, in 
realistic regimes, practical quantum protocols can achieve smaller information leakage than any 
classical protocol achieving the same task, hence showing the possibility of demonstrating a practical quantum advantage of three node quantum networks over their classical analogues.

\paragraph{Organization} In the next section, we state relevant definitions for communication complexity in the SMP model. We then 
define a notion of information leakage in this model, and provide significant evidence that it is the 
right one to consider.
In the following section, we state our technical lemmata and  combine 
them to prove our main result.
We then discuss the link between our results and the quantum fingerprinting protocol of Ref.~\cite{AL14} before concluding.
The proofs are relegated to the Appendix.

% SSSSSSSSSSSSSSSSSSSSSSSSSSSSSSSSSSSSSSSSSSSSSSSSSS %
\section{Preliminaries}
% SSSSSSSSSSSSSSSSSSSSSSSSSSSSSSSSSSSSSSSSSSSSSSSSSS %

We have the following definitions for the different simultaneous message passing (SMP) 
models of communication that we consider.
In all of these, $x \in X$ is Alice's input, $y \in Y$ is Bob's input, $r_A \in R_A$ is Alice's private 
randomness, $r_B \in R_B$ is Bob's private randomness, $r_C \in R_C$ is the referee's private 
randomness, $r_{AC} \in R_{AC}$ is the shared randomness between Alice and the referee, 
$r_{BC} \in R_{BC}$ is the shared randomness between Bob and the referee, $m_A \in M_A$ is 
Alice's message to the referee, $m_B \in M_B$ is Bob's message to the referee, and
$f : X \times Y \rightarrow Z$ is the function of $x$ and $y$ that Alice and Bob want the referee to compute.
Note that we abuse notation and overload the above notation for sets to also denote the corresponding random 
variables. We denote by $\D_{XY}$ the set of all joint probability distributions $\mu$ on $X \times Y$.
All logarithms are taken to base $2$, $e$ denotes the base of the natural logarithm, $\exp$ the exponential function in base $e$, and $\exp_2$ the exponential function in base $2$.

% SSSSSSSSSSSSSSSSSSSSSSSSSSSSSSSSSSSSSSSSSSSSSSSSSS %
\subsection{Private Coin}
% SSSSSSSSSSSSSSSSSSSSSSSSSSSSSSSSSSSSSSSSSSSSSSSSSS %

A protocol $\Pi$ in the private coin SMP model is defined by functions $\Pi_A: X \times R_A \rightarrow M_A$, 
$\Pi_B: Y \times R_B \rightarrow M_B$ and $\Pi_C : M_A \times M_B \times R_C \rightarrow Z$, and by distributions for the random strings $R_A$, $R_B$, $R_C$. 
We denote by $\Pi (x, y)$ the random variable on $Z$ corresponding to the output of the
 referee when Alice and Bob's inputs are $x$ and $y$, respectively, with the underlying distribution given by
the randomness used in $\Pi$, i.e.~$r_A, r_B$ and $ r_C$. 
The communication cost of protocol $\Pi$ is defined as 
\begin{align}
CC_{priv} (\Pi) = \lceil \log |M_A| \rceil + \lceil \log |M_B| \rceil.
\end{align}
The error of $\Pi$ for the function $f$ on input $(x, y)$ is defined as $P_e (\Pi,(x, y) ) = \mathrm{Pr}_\Pi[\Pi(x, y) \not= f(x, y)]$. The error of 
protocol $\Pi$ for the function $f$ is defined as $P_e (\Pi) = \max_{(x, y)} P_e (\Pi, (x, y))$. We denote the set of 
all protocols computing $f$ with error at most $\epsilon$ as $\T_{priv} (f, \epsilon)$. The communication complexity 
for computing $f$ with error $\epsilon$ is defined as 
\begin{align}
CC_{priv} (f, \epsilon) = \min_{\Pi \in \T_{priv} (f, \epsilon)} CC_{priv} (\Pi).
\end{align}

% SSSSSSSSSSSSSSSSSSSSSSSSSSSSSSSSSSSSSSSSSSSSSSSSSS %
\subsection{Shared Randomness}
% SSSSSSSSSSSSSSSSSSSSSSSSSSSSSSSSSSSSSSSSSSSSSSSSSS %

Similarly, a protocol in the shared randomness SMP model is 
defined by $\Pi_A: X \times R_A \times R_{AC} \rightarrow M_A$, 
$\Pi_B: Y \times R_B \times R_{BC} \rightarrow M_B$ and 
$\Pi_C : M_A \times M_B \times R_C \times R_{AC} \times R_{BC} \rightarrow Z$,
and by distributions for the random strings $R_A$, $R_B$, $R_C$, $R_{AC}$, $R_{BC}$.
Everything else is formally defined as in the private coin SMP 
model, with any $priv$ subscript replaced by a $sh$ subscript, 
and any averaging also accounting for $r_{AC}$ and $r_{BC}$. 
Note that we exclude the possibility of shared randomness between 
Alice and Bob, as in this case equality can be computed trivially in the SMP model.

% SSSSSSSSSSSSSSSSSSSSSSSSSSSSSSSSSSSSSSSSSSSSSSSSSS %
\subsection{Average Length}
% SSSSSSSSSSSSSSSSSSSSSSSSSSSSSSSSSSSSSSSSSSSSSSSSSS %

Protocols in the average length SMP model are defined as 
those in the shared randomness SMP model, but we also 
associate bit length functions $\ell_A : M_A \rightarrow \mathbb{N}$, 
$\ell_B : M_B \rightarrow \mathbb{N}$ on the message sets 
(these must satisfy some structural properties that correspond 
to a variable length encoding's ability to physically encode and decode information in the 
corresponding amount of bits, and in particular satisfy $\mathbb{E} (\ell_A (M_A)) \geq H (M_A)$, 
with $H(M_A)$ the Shannon entropy of random variable $M_A$,
and similarly for Bob's message. We do not further discuss these details since they are mostly irrelevant 
to the discussion here). 
A technical subtlety in the average length model is that we allow for message sets with messages of potentially unbounded length, and similarly we allow for potentially unbounded shared randomness, though we restrict our attention to finite expected message length, and correspondingly finite expected use of the shared randomness.
The average communication cost of  
protocol $\Pi$ on input $(x, y)$ is then defined  
as 
\begin{align}
CC_{av} (\Pi, (x, y)) = \mathbb{E}_{\Pi(x, y)} [\ell_A (M_A (x)) + \ell_B (M_B(y))], 
\end{align}
the average communication cost of protocol $\Pi$ as 
\begin{align}
CC_{av} (\Pi) = \max_{(x, y)} CC_{av} (\Pi, (x, y)), 
\end{align}
and the average communication complexity for computing $f$ with error $\epsilon$ as 
\begin{align}
CC_{av} (f, \epsilon) = \min_{\Pi \in \T_{av} (f, \epsilon)} CC_{av} (\Pi), 
\end{align}
with $\T_{av} $ the set of all protocols in the average length SMP 
model computing $f$ with $\epsilon$ error.

% SSSSSSSSSSSSSSSSSSSSSSSSSSSSSSSSSSSSSSSSSSSSSSSSSS %
\subsection{Link Between Complexities}
% SSSSSSSSSSSSSSSSSSSSSSSSSSSSSSSSSSSSSSSSSSSSSSSSSS %

We have the following chain of inequalities for the complexity of computing $f$ with $\epsilon$ error (we 
implicitly use the uniform length functions $\ell_A = \lceil \log |M_A| \rceil$, 
$\ell_B = \lceil \log |M_B| \rceil$ corresponding to equal length encodings 
for all messages of Alice and Bob, respectively, to link $CC_{sh}$ with $CC_{av}$):
\begin{align}
CC_{priv} (f, \epsilon) \geq CC_{sh} (f, \epsilon) \geq CC_{av} (f, \epsilon).
\end{align}
In Section~\ref{sec:main}, we prove inequalities in the reverse direction.

% SSSSSSSSSSSSSSSSSSSSSSSSSSSSSSSSSSSSSSSSSSSSSSSSSS %
\section{Information Leakage in the SMP Model}
\label{sec:ic_def}
% SSSSSSSSSSSSSSSSSSSSSSSSSSSSSSSSSSSSSSSSSSSSSSSSSS %

% SSSSSSSSSSSSSSSSSSSSSSSSSSSSSSSSSSSSSSSSSSSSSSSSSS %
\subsection{Information Leakage and Compression}
% SSSSSSSSSSSSSSSSSSSSSSSSSSSSSSSSSSSSSSSSSSSSSSSSSS %

In a communication complexity setting, the notion of information leakage (or information complexity) 
aims to quantify how much information the parties must reveal about their inputs to compute a given 
function. It is known that for general two-party interactive protocols, there can be an exponential 
gap between information and communication complexity of some functions~\cite{GKR14, GKR15}. However, for protocols 
with a bounded number of rounds, the two notions are known to be almost equivalent, up to some 
dependence on the number of rounds and the allowed increase in error~\cite{BR14, JPY12, BRWY13}.

Often, this equivalence is shown by first arguing, through a compression argument, about 
the distributional setting for which players want to achieve good average error for 
a fixed distribution on the inputs, and then using Yao's minimax theorem, which 
relates the distributional setting to the standard worst-case setting.

However, the analogue of Yao's minimax theorem does not hold in the SMP model, because 
we do not allow for shared randomness between Alice and Bob (otherwise, as pointed out earlier, computing 
equality becomes trivial). Nevertheless, Chabrakarti, Shi, Wirth and Yao~\cite{CSWY01} were able to 
show that for a large class of functions (including the equality function), information and 
communication complexities are related in the worst-case setting. Jain and Klauck~\cite{JK09} later
extended this result to all functions and relations. 
These results are not strong enough for our purposes: they are asymptotic in nature 
and the hidden constants are too large to allow us to show the separation that we seek 
relative to quantum information leakage. This is in part due to a limitation of their compression techniques; 
to obtain improved bounds, we use an alternate compression result due to Harsha, 
Jain, McAllister and Radhakrishnan~\cite{HJMR10}. Note that since Yao's minimax theorem does not 
hold for the SMP model, we must use the worst-case input version of the result of Ref.~\cite{HJMR10}. We state it here after the following definition. Here and 
throughout, $I(;)$ denotes the mutual information.

\begin{definition}
Let $\M_X : X \rightarrow \D_M$ be a noisy channel, i.e.~for each input $x$, we associate an output 
random variable $M_x$ over the set $M$.
For any $X \in \D_X$, denote by $XM$ the joint variable where the conditional 
probability distribution $M | X =x$ is distributed as $M_x$ for each $x$, and 
let $C_{\M} = \max_{X \in \D_X} I (X:M)$. An exact simulator for the channel $\M_X$ is a 
one-message protocol $\Pi$ in the average length model such that for any 
input $x \in X$ of Alice, the referee's output must be distributed exactly 
as $M_x$; see Figure~\ref{fig:channelsim}. Let $\T_{av} (\M_X)$ denote the set of all exact simulators for $\M_X$.
We define the communication complexity of $\M_X$ as
\begin{align*}
CC_{av} (\M_X) = \min_{\Pi \in \T_{av} (\M_X)} \max_{x \in X} CC_{av} (\Pi, x).
\end{align*} 
\end{definition}
\begin{lemma}[\cite{HJMR10}]
\label{lem:hjmrcomp}
For any $\M_X : X \rightarrow \D_M$, it holds that 
\begin{align}
CC_{av} (\M_X) \leq C_\M + g_1 (C_\M), 
\end{align}
with $g_1 (x) = 2 \log (x + 1)  + 10$.
\end{lemma}

\begin{figure}[t!]
\begin{center}
\includegraphics[width=\columnwidth]{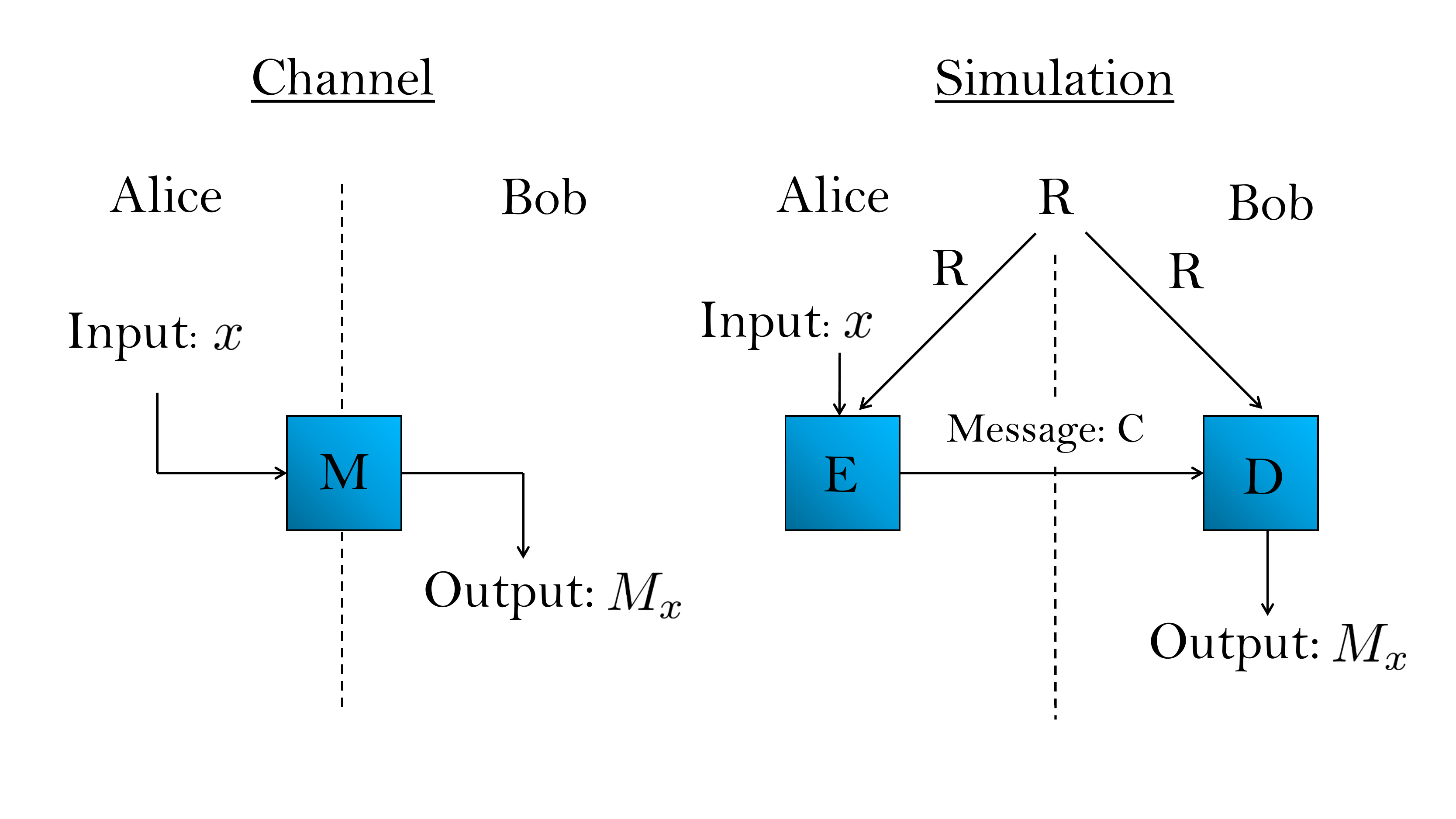}
\caption{Channel Simulation: the figure on the left hand side depicts a noisy channel with output distribution $M_x$ on input $x$, and the one on the right hand side, the simulation of such a channel. The goal of the simulation is to minimize the average length of the message $C$ while generating the same conditional output distribution $M_x$ on input $x$, while also being given free access to shared randomness $R$.}
\label{fig:channelsim}
\end{center}
\end{figure}

%\begin{figure}
%\centering
%\begin{minipage}{.5\textwidth}
%  \centering
%  \includegraphics[width=.7\linewidth]{channel.pdf}
%  \captionof{figure}{Noisy Channel}
%  \label{fig:test1}
%\end{minipage}%
%\begin{minipage}{.5\textwidth}
%  \centering
%  \includegraphics[width=.7\linewidth]{simulation.pdf}
%  \captionof{figure}{Channel Simulation, goal is to Minimize the length of message $C$}
%  \label{fig:test2}
%\end{minipage}
%\end{figure}

% SSSSSSSSSSSSSSSSSSSSSSSSSSSSSSSSSSSSSSSSSSSSSSSSSS %
\subsection{Definition of Information Leakage}
% SSSSSSSSSSSSSSSSSSSSSSSSSSSSSSSSSSSSSSSSSSSSSSSSSS %

We want our definition of information leakage in the SMP model to capture how much information the 
referee has at the end of the protocol about the joint input of Alice and Bob. In order 
to have a meaningful notion of information, we must first pass through a distributional 
definition; there is no information,  in an information-theoretic sense, to learn about a fixed input. 
The prior-free information leakage of a protocol is then defined by maximizing over all 
input distributions, in analogy to average communication cost. 

We define information leakage (and then information complexity) for protocols in 
the shared randomness SMP model.
However, the definition also applies for protocols in the average length SMP model, and, by 
taking $R_{AC}$ and $R_{BC}$ to be trivial registers, in the private randomness SMP model.
Furthermore, since $R_{AC}$ and $R_{BC}$ are independent of the inputs, note that Alice
and Bob can generate them using private randomness and transmit them to the 
referee without changing the information cost, though by increasing the communication.

\begin{definition}
 The information leakage of protocol $\Pi$ on input distribution $\mu$ is defined as \begin{align}
IL (\Pi, \mu) = I (XY; M_A M_B R_C R_{AC} R_{BC}),
\end{align}
and the information leakage of protocol $\Pi$ is defined as
\begin{align}
IL (\Pi) = \max_{\mu \in \D_{XY}} IL (\Pi, \mu).
\end{align}
The information leakage for computing $f$ with error $\epsilon$ is defined as
\begin{align}
IL (f, \epsilon) = \min_{\Pi \in \T_{sh} (f, \epsilon)} IL (\Pi).
\end{align}
\end{definition}

Note that $M_A M_B R_C R_{AC} R_{BC}$ is the set of registers available to the 
referee at the end of the protocol. However, since $R_{AC} R_{BC}$ do not contain 
information about $XY$, it follows from the chain rule for mutual information that 
\begin{align}
IL (\Pi, \mu) = I (XY;  M_A M_B R_C | R_{AC} R_{BC}).
\end{align}
 Moreover, $I(XY ; R_C | M_A M_B R_{AC} R_{BC}) = 0$ also, so we can further simplify as
\begin{align}
IL (\Pi, \mu) = I (XY; M_A M_B | R_{AC} R_{BC}).
\end{align}
Again, note that shared randomness can be replaced by private randomness without changing the information leakage by 
having Alice and Bob generate the shared randomness privately and then transmitting it at no information 
cost. Allowing for shared randomness however allows us to link information and average 
communication through compression arguments.

% SSSSSSSSSSSSSSSSSSSSSSSSSSSSSSSSSSSSSSSSSSSSSSSSSS %
\subsection{Alternate Characterizations of Information Leakage}
% SSSSSSSSSSSSSSSSSSSSSSSSSSSSSSSSSSSSSSSSSSSSSSSSSS %

Another notion of information is also natural to consider in the SMP model, and satisfies many interesting properties. Following Chakrabarti, Shi, Wirth, and Yao~\cite{CSWY01}, we call it information complexity.
\begin{definition}
 The information complexity of protocol $\Pi$ on input distribution $\mu$ is defined as \begin{align}
IC (\Pi, \mu) = I (X; M_A | R_{AC}) + I (Y; M_B | R_{BC}),
\end{align}
and the information complexity of protocol $\Pi$ is defined as
\begin{align}
IC (\Pi) = \max_{\mu \in \D_{XY}} IC (\Pi, \mu).
\end{align}
The information complexity for computing $f$ with error $\epsilon$ is defined as
\begin{align}
IC (f, \epsilon) = \min_{\Pi \in \T_{sh} (f, \epsilon)} IC (\Pi).
\end{align}
\end{definition}

Note that for any $\Pi$, $\mu$, $f$ and $\epsilon$, it holds that
\begin{align}
IC (\Pi, \mu)& \leq \mathbb{E}_\mu [CC_{av} (\Pi, (x, y))], \\
IC (\Pi) & \leq CC_{av} (\Pi), \\
IC (f, \epsilon) &\leq CC_{av} (f, \epsilon),
\end{align}
in which we used that $I(X ; M_A | R_{AC}) \leq H (M_A) \leq \bE_{(x, y) \sim \mu, \Pi(x, y)} [\ell_A (M_A (x))]$, and similarly for Bob's message.
The information complexity satisfies an additivity property, is continuous in the error parameter for $\epsilon > 0$, and, up to a small additive logarithmic term, is equivalent to the average communication complexity, a result that follows from two applications of the compression result in Lemma~\ref{lem:hjmrcomp}. 
\begin{lemma}
\label{lem:icvsccav}
For any $f$ and any $\epsilon \in [0, \frac{1}{2})$, it holds that 
\begin{align}
%CC_{av} (\Pi) &\leq IC (\Pi) + O (\log (IC (\Pi) + 1)), \\
CC_{av} (f, \epsilon) &\leq IC (f, \epsilon) + 2 g_1 (IC(f, \epsilon)),
\end{align}
with $g_1 (x) = 2 \log (x + 1) + 10$.
\end{lemma}

We provide a proof in Appendix~\ref{app:comp}.
Note that $\log IC (f, \epsilon) \leq \log CC_{av}(f, \epsilon)$. 
%Note that the distributional version of the above result, $CC_{av} (\Pi, \mu) \leq IC (\Pi, \mu) + O (\log (IC (\Pi, \mu) + 1))$, also holds, but this is not strong enough for our purposes.

These properties of information complexity imply that it is equal to the amortized communication complexity ($ACC$), i.e.~the optimal asymptotic average length communication complexity per copy for solving many copies of the same function in parallel.
\begin{theorem}
For any $f$ and $\epsilon \in (0, \frac{1}{2})$, it holds that 
\begin{align} 
IC (f ,\epsilon) = ACC (f, \epsilon).
\end{align}
\end{theorem}

Information leakage and information complexity are closely related:
\begin{lemma}
\label{lem:icvsildist}
For any $\Pi$ and $\mu$, 
\begin{align}
IL(\Pi, \mu) &\leq IC (\Pi, \mu) \\
& = IL(\Pi, \mu) + I(M_A ; M_B| R_{AC} R_{BC})\\
	& \leq 2 IL (\Pi, \mu).
\end{align}
\end{lemma}

Perhaps surprisingly, we can avoid the factor of $2$ for worst-case input, an important fact for our practical application.
\begin{lemma}
\label{lem:icvsilworst}
For any $\Pi$, $f$ and $\epsilon \in [0, \frac{1}{2})$,
\begin{align}
IC (\Pi) &= IL (\Pi), \\
IC (f, \epsilon) &= IL(f, \epsilon).
\end{align}
\end{lemma}

Hence, the information leakage, being equal to information complexity, also corresponds to the amortized communication complexity, providing an operational interpretation for it. Proofs for Lemmata~\ref{lem:icvsildist} and~\ref{lem:icvsilworst} are provided in Appendix~\ref{app:ilvsic}.

% SSSSSSSSSSSSSSSSSSSSSSSSSSSSSSSSSSSSSSSSSSSSSSSSSS %
\section{Information Leakage Lower Bound}
\label{sec:main}
% SSSSSSSSSSSSSSSSSSSSSSSSSSSSSSSSSSSSSSSSSSSSSSSSSS %

We show the following lower bound on the information leakage for computing the equality function  ($EQ_n$)
on $n$ bits in the SMP model.

\begin{theorem}
	For any $n$,  any $\epsilon \geq 0$, any $\delta_1 > 0$, and any $\delta_2 > 0$ satisfying $\epsilon+\delta_1+\delta_2<\frac{1}{2}$, it holds that 
\begin{align*}
IL (EQ_n, \epsilon) \geq  \delta_1 \left( 2 \sqrt{g_3 (\epsilon + \delta_1 +\delta_2)} \sqrt{n} - g_3 (\epsilon + \delta_1 +\delta_2)  - g_2 (n, n, \delta_2) - 10  \right) - 2 g_1 (2n),
\end{align*}
with $g_1 (x) = 2 \log (x + 1) + 10$,
$g_2 (x, y, z) = 2 \log (\frac{2 (x + y)}{z^2 \cdot \log e } +1) + 2$, and
$g_3 (x) = 2 \cdot (1/2 - x)^2 \cdot  \log e $.
\end{theorem}

% SSSSSSSSSSSSSSSSSSSSSSSSSSSSSSSSSSSSSSSSSSSSSSSSSS %
\subsection{Sketch of Proof}
% SSSSSSSSSSSSSSSSSSSSSSSSSSSSSSSSSSSSSSSSSSSSSSSSSS %

The high-level idea can be split into two parts as follows. On one side, we show
a parameterized lower bound on the communication complexity of $EQ_n$ in 
the private coin SMP model in terms of the allowed 
worst-case error $\epsilon$.
% and the $1$-way deterministic communication complexity $(CC_{1})$. 
%For the equality function $(EQ_n)$ on $n$-bit, the deterministic communication 
%complexity is maximal: $n$ bits for $n$-bit inputs. 
On the other side, we show a general link between $IL$ and $CC_{priv}$ of Boolean functions through a series of reductions.
First, we can relate $CC_{priv}$ and $CC_{sh}$ using a Newman type 
argument~\cite{New}. This only incurs a manageable additive loss.
We then use a Markov inequality argument to relate $CC_{sh}$  and $CC_{av}$.
This is the most costly reduction.
The equivalence, up to a small additive logarithmic factor, of $IL$ and $CC_{av}$ then completes the argument.

% SSSSSSSSSSSSSSSSSSSSSSSSSSSSSSSSSSSSSSSSSSSSSSSSSS %
\subsection{Statement of Lemmata}
% SSSSSSSSSSSSSSSSSSSSSSSSSSSSSSSSSSSSSSSSSSSSSSSSSS %

We obtain our  main result by combining the following lemmata.
Their proofs are relegated to Appendix~\ref{app:icvscc}.

We show a variation of Newman's theorem, relating $CC_{priv}$ and $CC_{sh}$. 
Jain and Klauck~\cite{JK09} noted that such a result holds for SMP models; to obtain better bounds, 
we adapt the proofs from Ref.~\cite{Pit14, Bra11} to show:

\begin{lemma}
\label{lem:ccprivvssh}
For any Boolean function $f$, any $\epsilon \geq 0$, and any $\delta > 0$ satisfying $\epsilon + \delta < \frac{1}{2}$, denote $n_A = \log |X|$ and $n_B = \log |Y|$. Then it holds that
\begin{align*}
CC_{priv} (f, \epsilon + \delta) \leq CC_{sh} (f, \epsilon) +  g_2 (n_A, n_B, \delta),
\end{align*}
with $g_2 (x, y, z) = 2 \log (\frac{2 (x + y)}{z^2 \cdot \log e} +1) + 2$.
\end{lemma}

A Markov inequality argument allows us to relate $CC_{sh}$ and $CC_{av}$.

\begin{lemma}
\label{lem:ccshvsav}
For any Boolean function $f$, any $\epsilon \geq 0$, and any $\delta > 0$ satisfying $\epsilon + \delta < \frac{1}{2}$, 
\begin{align*}
CC_{sh} (f, \epsilon + \delta) \leq (1 / \delta) \cdot CC_{av} (f, \epsilon) + 4.
\end{align*}
\end{lemma}

We can  link  $IL$ and $CC_{priv}$ by combining the above two results and the link 
between $IL$ and $CC_{av}$ proven in Section~\ref{sec:ic_def}.
 
\begin{lemma}
\label{lem:ilvsccpriv}
For any Boolean function $f$, any $\epsilon \geq 0$, any $\delta_1 > 0$, and  any $\delta_2 > 0$ satisfying $\epsilon + \delta_1 + \delta_2 < \frac{1}{2}$, denote $n_A = \log |X|$ and $n_B = \log |Y|$. Then it holds that 
\begin{align*}
IL (f, \epsilon) \geq  \delta_1 \big( CC_{priv} (f, \epsilon + \delta_1 +\delta_2) - g_2 (n_A, n_B, \delta_2) - 4  \big) - 2 g_1 (CC_{priv} (f, \epsilon)),
\end{align*}
with $g_1 (x) = 2 \log (x + 1) + 10$, and
$g_2 (x, y, z) = 2 \log (\frac{2 ( x + y)}{z^2 \cdot \log e } +1) + 2$.
\end{lemma}

We have the following lower bound on $CC_{priv} (EQ_n, \epsilon)$,
% relating $CC_{1}$, the one-way deterministic communication complexity, and $CC_{priv}$, 
by adapting a simplification of an argument from Babai and Kimmel~\cite{BK97} due to Gavinsky, Regev and de Wolf~\cite{GRW08}.

\begin{lemma}
\label{lem:bbkim}
For any $n \in \mathbb{N}$ and any $\epsilon \in [0, \frac{1}{2} )$, the following holds:
\begin{align*}
CC_{priv} (EQ_n, \epsilon) \geq 2 \sqrt{g_3 (\epsilon)} \sqrt{n} - g_3 (\epsilon) - 6,
\end{align*}
with $g_3 (x) = 2 \cdot (1/2 - x)^2 \cdot \log e $.
\end{lemma}

In order to compare with the result from Babai and  Kimmel, we can take $\epsilon = 0.01$ and large enough input such that our bound essentially gives 
\begin{align}
CC_{priv} (EQ_n, 0.01) \geq 1.66 \sqrt{n},
\end{align}
in contrast to 
\begin{align}
CC_{priv} (EQ_n, 0.01) \geq 0.1 \sqrt{n},
\end{align}
an improvement by more than an order of magnitude.

Inserting the result of Lemma~\ref{lem:bbkim} into Lemma~\ref{lem:ilvsccpriv}, 
and also using the trivial bound $CC_{priv} (EQ_n, \epsilon) \leq 2n$, gives us the desired bound, that is
\begin{align*}
IL (EQ_n, \epsilon) \geq  \delta_1 \big( 2 \sqrt{g_3 (\epsilon + \delta_1 +\delta_2)} \sqrt{n} - g_3 (\epsilon + \delta_1 +\delta_2)  - g_2 (n, n, \delta_2) - 10 \big) - 2 g_1 (2n).
\end{align*}

% SSSSSSSSSSSSSSSSSSSSSSSSSSSSSSSSSSSSSSSSSSSSSSSSSS %
\section{Connection with Experimental Quantum Fingerprinting}
% SSSSSSSSSSSSSSSSSSSSSSSSSSSSSSSSSSSSSSSSSSSSSSSSSS %

Recent theoretical and experimental advances have led to the possibility of demonstrating quantum fingerprinting protocols that are capable of beating the classical communication lower bound for equality in the SMP model. In this section, we argue that these practical protocols can also surpass the lower bound on classical information leakage that we have derived in this work.

Practical quantum fingerprinting protocols are based on coherent states of light. For fixed input size $n$ and error probability $\epsilon$, it is possible, by introducing an arbitrarily small additional error, to effectively make the protocol operate in an Hilbert space of dimension equivalent to one of $O(\mu \log n)$ qubits, where $\mu$ is the total mean photon number. Thus, since the dimension of the signals gives an upper bound on the information leakage, the quantum information leakage (QIL) of such practical quantum fingerprinting protocols satisfies

\begin{align}
QIL=O(\mu \log n).
\end{align}

The precise expression for the upper bound on the information leakage of this quantum fingerprinting protocol can be found in Ref.~\cite{AL14}.

\begin{figure}[h!]
\begin{center}
\includegraphics[width=0.8\columnwidth]{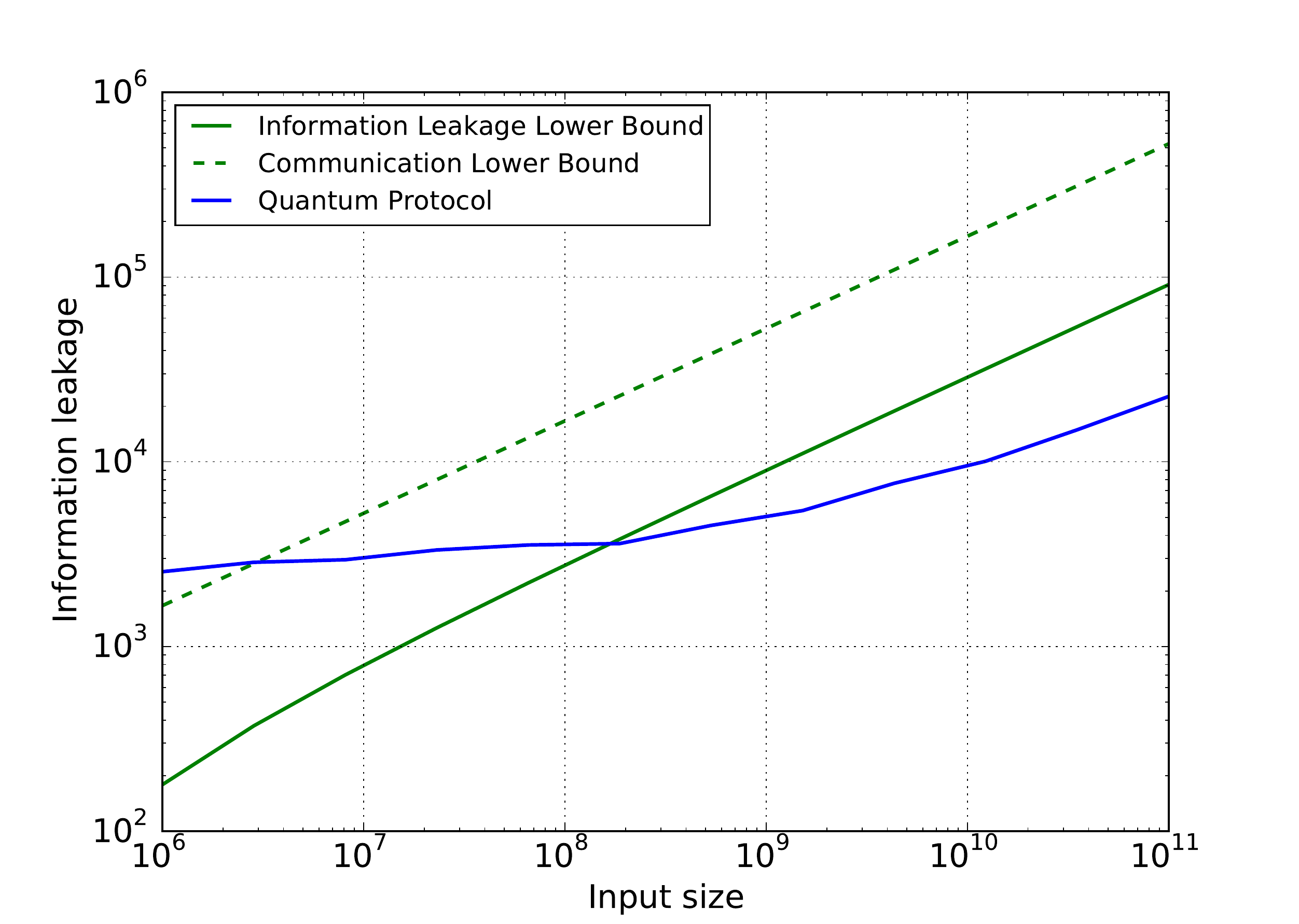}
\caption{Information leakage for practical quantum fingerprinting protocols compared to classical lower bounds. As experimental parameters, we have assumed a visibility of $\nu=0.98$, dark count rate of $0.11 Hz$, transmissivity of $\eta=0.3$ and probability of error $\epsilon=0.01$. }\label{Plot}
\end{center}
\end{figure}

In Fig. \ref{Plot}, we plot the upper bound on the information leakage of quantum fingerprinting  for realistic experimental parameters and compare it to our classical information leakage lower bound and the classical communication lower bound. The information leakage lower bound was optimized over $\delta_1$ and $\delta_2$ under the constraint that $\delta_1+\delta_2+\epsilon<\frac{1}{2}$. Even though the information leakage lower bound is significantly smaller than the communication lower bound, quantum protocols can still operate in a regime where they leak less information than any classical protocol.

% SSSSSSSSSSSSSSSSSSSSSSSSSSSSSSSSSSSSSSSSSSSSSSSSSS %
\section{Conclusion}
% SSSSSSSSSSSSSSSSSSSSSSSSSSSSSSSSSSSSSSSSSSSSSSSSSS %

We proved lower bounds on the information leakage of any classical 
protocol computing the equality function on $n$ bits in the SMP model.
Asymptotic bounds scaling as $\Omega ( \sqrt{n} )$ were already known, 
but the hidden constants were too large for the goal we have in mind: 
obtaining finite size bounds in a realistic regime accessible to practical 
quantum fingerprinting protocols, in order to show an advantage for 
such quantum protocols over any classical protocol. We succeeded in 
this endeavour, and the next step is thus to perform an experiment 
demonstrating this quantum advantage in terms of information leakage. 
Since practical quantum protocols within reach of current technology 
achieve low information transmission through protocols using a large number of signals, such a comparison in terms of classical information leakage 
instead of classical communication cost should allow for an actual experimental
demonstration of a quantum advantage.

% SSSSSSSSSSSSSSSSSSSSSSSSSSSSSSSSSSSSSSSSSSSSSSSSSS %
\section*{Acknowledgements}
% SSSSSSSSSSSSSSSSSSSSSSSSSSSSSSSSSSSSSSSSSSSSSSSSSS %

The authors are grateful to Norbert L\"utkenhaus for useful discussions and feedback on this manuscript, and
to Harry Buhrman for the idea to investigate the information leakage advantage of practical quantum fingerprinting protocols.

JMA acknowledges support from the Mike and Ophelia Lazaridis Fellowship, the Singapore Ministry of Education (partly
through the Academic Research Fund Tier 3 MOE2012-
T3-1-009) and the National Research Foundation of Singapore,
Prime Minister’s Office, under the Research Centres
of Excellence programme. 
DT was supported in part by NSERC, Industry Canada and ARL CDQI program. 
IQC and
PI are supported in part by the Government of Canada and the Province of
Ontario.

\appendix

% SSSSSSSSSSSSSSSSSSSSSSSSSSSSSSSSSSSSSSSSSSSSSSSSSS %
\section{Proofs of Compression Lemma}
\label{app:comp}
% SSSSSSSSSSSSSSSSSSSSSSSSSSSSSSSSSSSSSSSSSSSSSSSSSS %

We restate Lemma~\ref{lem:icvsccav} here for convenience.

\begin{lemma}
For any $f$ and any $\epsilon \in [0, \frac{1}{2})$, it holds that 
\begin{align}
%CC_{av} (\Pi) &\leq IC (\Pi) + O (\log (IC (\Pi) + 1)), \\
CC_{av} (f, \epsilon) &\leq IC (f, \epsilon) + 2 g_1 (IC(f, \epsilon)),
\end{align}
with $g_1 (x) = 2 \log (x + 1) + 10$.

\end{lemma}

\begin{proof}
We show that for any protocol $\Pi$ in the shared randomness SMP model, there exists a simulation 
protocol $\Pi^\prime$ in the average length SMP model that exactly simulates $\Pi$, in the sense 
that the referee can compute $\tilde{M}_A$ from $M_A^\prime R_{AC}^\prime$ and $\tilde{M}_B$ 
from $M_B^\prime R_{BC}^\prime$ such that for any $x, y$, $\tilde{M}_A \tilde{M}_B R_C| X=x, Y=y$ 
is distributed exactly as $M_A M_B R_{AC} R_{BC} R_C | X=x, Y=y$, hence he can then 
compute a protocol output $\Pi^\prime (x, y)$ distributed exactly as $\Pi (x, y)$.
Moreover, $CC_{av} (\Pi^\prime) \leq IC(\Pi) + 2 g_1 (IC (\Pi))$, and the result follows.

We define the protocol $\Pi^\prime$ in the following way, viewing $R_{AC} M_A$ as a noisy 
channel $\M_A^\prime$ with input $x \in X$ and output set $R_{AC} \times M_{A}$, and 
similarly for $R_{BC} M_B$ with input $y \in Y$.

Then $R_A^\prime$ is empty and $R_{AC}^\prime$ consists of random strings required for 
the simulator of Lemma~\ref{lem:hjmrcomp}, for the channel $\M_A^\prime$. Alice sends as 
her message $M_A^\prime$ the message required by this simulator, 
and the referee generates a virtual register $\widetilde{M}_A$ from $R_{AC}^\prime M_A^\prime$ 
such that for any $x$, the virtual register on $\widetilde{M}_A$ in $\Pi^\prime$ is distributed 
exactly as the registers on  $R_{AC} M_A$ in $\Pi$. Then, for any $x$ (and for any $y$),
\begin{align*}
\mathbb{E}_{\Pi^\prime (x, y)} [\ell_A (M_A^\prime (x))] & = CC_{av} (\M_A^\prime) \\
				& \leq C_{\M_A^\prime} +   g_1 (C_{\M_A^\prime}).
\end{align*}
Note that $C_{\M_A^\prime} = \max_{X \in \D_X} I(X ; R_{AC} M_A) = \max_{X \in \D_X} I(X ; M_A| R_{AC})$.

We similarly define everything on Bob's side and get $C_{\M_A^\prime} + C_{\M_B^\prime} = IC (\Pi)$. 
Also define $R_C^\prime = R_C$ and 
%the referee's function 
$\Pi_C^\prime = \Pi_C$, so that indeed for any input $(x, y)$, $\widetilde{M}_A \widetilde{M}_B R_C | (X= x, Y=y)$ 
is distributed as $M_A M_B R_{AC} R_{BC} R_C | (X=x, Y=y)$, hence $\Pi^\prime (x, y)$ is distributed as $ \Pi (x, y)$, and 
\begin{align*}
CC_{av} (\Pi^\prime) & = \max_{(x, y)} \mathbb{E}_{\Pi^\prime (x, y)} [\ell_A (M_A^\prime (x)) + \ell_B (M_B^\prime (y))] \\
		& = CC_{av} (\M_A^\prime) + CC_{av} (\M_B^\prime) \\
		& \leq IC (\Pi) + 2 g_1 (IC (\Pi)),
\end{align*}
and the result follows.
\end{proof}

% SSSSSSSSSSSSSSSSSSSSSSSSSSSSSSSSSSSSSSSSSSSSSSSSSS %
\section{Proofs of Lemmata relating $IL$ and $IC$}
\label{app:ilvsic}
% SSSSSSSSSSSSSSSSSSSSSSSSSSSSSSSSSSSSSSSSSSSSSSSSSS %

We restate Lemma~\ref{lem:icvsildist} and~\ref{lem:icvsilworst} here for convenience.

% SSSSSSSSSSSSSSSSSSSSSSSSSSSSSSSSSSSSSSSSSSSSSSSSSS %
\subsection{Distributional inputs}
% SSSSSSSSSSSSSSSSSSSSSSSSSSSSSSSSSSSSSSSSSSSSSSSSSS %

\begin{lemma}
For any $\Pi$ and $\mu$, 
\begin{align}
IL(\Pi, \mu) &\leq IC (\Pi, \mu) \\
& = IL(\Pi, \mu) + I(M_A ; M_B| R_{AC} R_{BC})\\
	& \leq 2 IL (\Pi, \mu).
\end{align}
\end{lemma}

\begin{proof}
We first show that $IC(\Pi, \mu) = IL(\Pi, \mu) + I(M_A ; M_B | R_{AC} R_{BC})$ by the following chain of inequality:
\begin{align*}
IL (\Pi, \mu) = & I(XY; M_A M_B | R_{AC} R_{BC}) \\
			  = & I(XY ; M_A | R_{AC} R_{BC}) + I(XY ; M_B | M_A R_{AC} R_{BC}) \\
			 = & I(X; M_A | R_{AC} R_{BC}) + I (Y ; M_A | X R_{AC} R_{BC}) \\
				& +  I(Y ; M_B | M_A R_{AC} R_{BC}) + I (X ; M_B | Y M_A R_{AC} R_{BC}) \\
			 = & I(X ; M_A | R_{AC}) + I (Y M_A ; M_B | R_{AC} R_{BC}) - I (M_A ; M_B | R_{AC} R_{BC}) \\
			 = & I(X ; M_A | R_{AC}) + I (Y  ; M_B | R_{AC} R_{BC}) \\
				& + I(M_A ; M_B | Y R_{AC} R_{BC}) - I (M_A ; M_B | R_{AC} R_{BC}) \\
			 = & I(X ; M_A | R_{AC}) + I (Y  ; M_B | R_{BC})  - I (M_A ; M_B | R_{AC} R_{BC}) \\
			 = & IC (\Pi, \mu) - I (M_A ; M_B | R_{AC} R_{BC}).
\end{align*}
Along with the chain rule, we made use of the fact that $R_{BC}$ is independent of $X   M_A R_{AC}$ 
and $R_{AC}$ of $Y   M_B  R_{BC}$, as well as the fact that the following are short Markov 
chains: $Y \leftrightarrow X R_{AC} R_{BC} \leftrightarrow M_A$, 
$X \leftrightarrow Y M_A R_{AC} R_{BC} \leftrightarrow M_B$,  and
$M_A \leftrightarrow  Y R_{AC} R_{BC} \leftrightarrow M_B$.

The remaining inequalities follow by non-negativity of mutual information and by the following chain of inequality:
\begin{align*}
I(M_A ; M_B | R_{AC} R_{BC}) & = I (Y M_A ; M_B | R_{AC } R_{BC}) - I(Y ; M_B | M_A R_{AC} R_{BC}) \\
				& \leq I(Y ; M_B | R_{AC} R_{BC}) + I (M_A ; M_B | Y R_{AC} R_{BC}) \\
				& = I (Y ; M_B | R_{AC} R_{BC}) \\
				& \leq I(XY; M_A M_B | R_{AC} R_{BC}) \\
				&  = IL (\Pi, \mu).
\end{align*}
\end{proof}

% SSSSSSSSSSSSSSSSSSSSSSSSSSSSSSSSSSSSSSSSSSSSSSSSSS %
\subsection{Worst-case inputs}
% SSSSSSSSSSSSSSSSSSSSSSSSSSSSSSSSSSSSSSSSSSSSSSSSSS %

\begin{lemma}
For any $\Pi$, $f$ and $\epsilon \in [0, \frac{1}{2})$,
\begin{align}
IC (\Pi) &= IL (\Pi), \\
IC (f, \epsilon) &= IL(f, \epsilon).
\end{align}
\end{lemma}

\begin{proof}
On one side, $\max_{\mu \in \D_{XY}} IL (\Pi, \mu) \leq \max_{\mu \in \D_{XY}} IC (\Pi, \mu)$ 
follows by Lemma~\ref{lem:icvsildist}. On the other side, by restricting the maximization over $\D_{XY}$ 
only to product distributions $\mu_X \otimes \mu_Y$, we get
\begin{align*}
IL (\Pi) & =\max_{\mu \in \D_{XY}} IL (\Pi, \mu) \\
			& \geq \max_{\mu_X \otimes \mu_Y} IL (\Pi, \mu) \\
			& = \max_{\mu_X \otimes \mu_Y} I (XY; M_A M_B |R_{AC} R_{BC}) \\
			& = \max_{\mu_X \otimes \mu_Y} (I (X; M_A  |R_{AC} ) + I (Y ; M_B | R_{BC})) \\
			& = \max_{\mu_X} I (X; M_A  |R_{AC} )  + \max_{ \mu_Y}  I (Y ; M_B | R_{BC}) \\
			& = \max_{\mu \in \D_{XY}} I (X; M_A  |R_{AC} )  + \max_{ \mu \in \D_{XY}}  I (Y ; M_B | R_{BC}) \\
			& \geq \max_{\mu \in \D_{XY}} ( I (X; M_A  |R_{AC} )  + I (Y ; M_B | R_{BC})) \\
				& = IC (\Pi),
\end{align*}
in which we use the fact that $X M_A R_{AC}$ and $Y M_B R_{BC}$ are independent if $X$ and $Y$ are.
The results follow.
\end{proof}

% SSSSSSSSSSSSSSSSSSSSSSSSSSSSSSSSSSSSSSSSSSSSSSSSSS %
\section{Proofs of Lemmata relating $IC$ and $CC_{priv}$}
\label{app:icvscc}
% SSSSSSSSSSSSSSSSSSSSSSSSSSSSSSSSSSSSSSSSSSSSSSSSSS %

For convenience, we restate the lemmata before their proofs.

% SSSSSSSSSSSSSSSSSSSSSSSSSSSSSSSSSSSSSSSSSSSSSSSSSS %
\subsection{Link between $CC_{priv}$ and $CC_{sh}$ (Lemma~\ref{lem:ccprivvssh})}
% SSSSSSSSSSSSSSSSSSSSSSSSSSSSSSSSSSSSSSSSSSSSSSSSSS %

\begin{lemma}
For any Boolean function $f$, any $\epsilon \geq 0$, and any $\delta > 0$ satisfying $\epsilon + \delta < \frac{1}{2}$, denote $n_A = \log |X|$ and $n_B = \log |Y|$. Then it holds that
\begin{align*}
CC_{priv} (f, \epsilon + \delta) \leq CC_{sh} (f, \epsilon) +  g_2 (n_A, n_B, \delta),
\end{align*}
with $g_2 (x, y, z) = 2 \log (\frac{2 ( x + y)}{z^2 \cdot \log e } +1) + 2$.
\end{lemma}

\begin{proof}
We use Hoeffding's inequality: Let $x_1, \cdots, x_t$ be i.i.d random variables in $[0, 1]$, 
and denote the empirical mean $\frac{1}{t}\sum_{i=1}^t x_i = \bar{x}$. 
Then 
\begin{align*}
\Pr [\bar{x} - \mathbb{E} [\bar{x}] \geq \delta] & \leq \exp{(-2 \delta^2 t)} \\
							& = \exp_2 (-2 \delta^2 t \cdot  \log e ).
\end{align*}

%We denote $n_A = \lceil \log |X| \rceil$, $n_B = \lceil \log |Y| \rceil$, and $t = \lceil \frac{n_A n_B}{\delta^2} \rceil +1$.
Given $\delta > 0$ and $\delta_1 > 0$, we take $t = \lceil \frac{n_A + n_B}{2 \delta^2 \cdot \log e } \rceil + \delta_1$.

Given a protocol $\Pi$ in the shared randomness SMP model, fix inputs $(x, y)$ and random 
strings $r_{ABC} = (r_A, r_B, r_C, r_{AC}, r_{BC})$ used for one run of the 
protocol, and let $\Pi (x, y, r_{ABC})$ be the output of the protocol when run 
on input $x, y$ and random strings $r_{ABC}$.
% = h_C (g_A (x, r_A, r_{AC}), g_B (y, r_B, r_{BC}), r_C, r_{AC}, r_{BC})))$ be the output of the protocol on these. 
Let $E (x, y, r_{ABC}) = 1$ if $\Pi (x, y, r_{ABC}) \not= f(x, y)$, and $0$ otherwise. 
For all $x \in X$ and $ y \in Y $, it holds that Pr$_{r_{ABC}} [ E(x, y, r_{ABC}) = 1] \leq \epsilon$,
since $\Pi$ makes error at most $\epsilon$. 
We now start by partially derandomizing the protocol on Alice's side. Fix $t$ random 
strings $r_{AC}^1, \cdots r_{AC}^t \in R_{AC}$, and denote by $E (x, y, r_{AC}^i)$ 
the average of $E (x, y, r_{ABC})$ over $r_A, r_B, r_C, r_{BC}$ obtained
while fixing $r_{AC} = r_{AC}^i$. Then also $\mathbb{E}_{r_{AC}^i \sim R_{AC}} [ E(x, y, r_{AC}^i) ] \leq \epsilon$ 
for any $i$. Denoting $\bar{E} (x, y, r_{AC}^i) = \frac{1}{t}\sum_{i = 1}^{t} E (x, y, r_{AC}^i)$, we get by 
Hoeffding's inequality that 
Pr$_{r_{AC}^1, \cdots r_{AC}^t}[\bar{E} (x, y, r_{AC}^i) - \epsilon \geq  \delta] \leq \exp{(-2 \delta^2 t)}$. 
With the above choice for $t$, we get that for any pair $(x, y)$ of 
input,  Pr$_{r_{AC}^1, \cdots r_{AC}^t}[\bar{E} (x, y, r_{AC}^i) \geq (\epsilon + \delta)] < \exp_2 (- (n_A + n_B))$. 
By the union bound, there exists a choice of the $r_{AC}^i$ such that $ \bar{E} (x, y, r_{AC}^i)   < (\epsilon + \delta)$ 
for all $x$ and $ y$. Let $\widetilde{\Pi}$ be the protocol in which Alice and the referee agree beforehand 
on the set $\{r_{AC}^i \}_{i = 1}^t$ and they pick $i \in [t]$ 
uniformly at random before running $\Pi$ with $r_{AC} = r_{AC}^i$. Then $\widetilde{\Pi}$ 
has worst-case error at most $\epsilon + \delta$. 
Starting with protocol $\widetilde{\Pi}$, we can do similarly on Bob's side, and partially derandomize $r_{BC}$ by finding 
a set $\{r_{BC}^j \}_{j=1}^t$ of size $t$ such that the error of a protocol $\widetilde{\Pi}^\prime$ in which Bob 
and the referee pick $j \in [t]$ 
uniformly at random before running the protocol $\widetilde{\Pi}$, with $r_{AC}$ also partially derandomized, has error at 
most $\epsilon + 2 \delta$. 
In order to obtain a private coin protocol $\Pi^\prime$, we instead have Alice and Bob pick locally $i$ and $j$ uniformly at random before transmitting 
it to the referee as a prefix to their message, so that the referee is also aware of the choice of $r_{AC}^i$ 
and $r_{BC}^i$, and they can then run  $\widetilde{\Pi}^\prime$. 
It then holds that
\begin{align*}
CC_{priv} (\Pi^\prime) & \leq CC_{sh} (\Pi) + 2 \lceil \log t \rceil \\
	& \leq CC_{sh} (\Pi) + 2 \log (\frac{n_A + n_B}{ 2 \delta^2 \cdot \log e } + 1 + \delta_1) + 2.
\end{align*}
We get the desired  result by taking $\delta_1$ to zero and by considering $\delta^\prime = 2 \delta$.
\end{proof}

% SSSSSSSSSSSSSSSSSSSSSSSSSSSSSSSSSSSSSSSSSSSSSSSSSS %
\subsection{Link between $CC_{sh}$ and $CC_{av}$ (Lemma~\ref{lem:ccshvsav})}
% SSSSSSSSSSSSSSSSSSSSSSSSSSSSSSSSSSSSSSSSSSSSSSSSSS %

\begin{lemma}
For any Boolean function $f$, any $\epsilon \geq 0$, and any $\delta > 0$ satisfying $\epsilon + \delta < \frac{1}{2}$, 
\begin{align*}
CC_{sh} (f, \epsilon + \delta) \leq (1 / \delta) \cdot CC_{av} (f, \epsilon) + 4.
\end{align*}
\end{lemma}

\begin{proof}
We make use of Markov's inequality, stating, for a non-negative random variable $X$ and $a>0$, that 
\begin{align*}
\Pr [X \geq a] \leq \frac{\mathbb{E (X)}}{a}.
\end{align*}

Given protocol $\Pi$ in the average length SMP model computing $f$ with error $\epsilon$ and satisfying $CC_{av} (\Pi) = CC_{av} (f, \epsilon)$, 
we will construct protocol $\Pi^\prime$ in the shared randomness SMP model with $CC_{sh} (\Pi^\prime) \leq \frac{1}{\delta} CC_{av} (\Pi) + 2$ and 
simulating $\Pi$ exactly on short enough messages, and aborting on too long messages.

Let  $c_A (x, y) = \bE_{\Pi (x, y)} [\ell_A (M_A (x))]$, $c_B (x, y) = \bE_{\Pi (x, y)} [\ell_B (M_B (y))]$ 
and $c = \max_{(x, y)} (c_A (x, y) + c_B (x, y)) = CC_{av} (\Pi)$. Fix $\delta > 0$, then, for any $x, y $, 
Markov's inequality implies that Pr$[\ell_A (M_A (x)) \geq \frac{c_A}{\delta}] \leq \delta$ 
and Pr$[\ell_B (M_B (y)) \geq \frac{c_B}{\delta}] \leq \delta$. Define $\Pi^\prime$ on Alice's side
as running $\Pi$ and sending the message $m_A$ if it is of length at most $\frac{c_A}{\delta}$, 
and sending an abort flag otherwise, so 
that the referee can output a random answer in such a case.
Define $\Pi^\prime$ similarly on Bob's side, depending whether $m_B$  is of length at most $\frac{c_B}{\delta}$.
It holds 
that $CC_{sh} (\Pi^\prime) \leq \frac{1}{\delta} CC_{av} (\Pi) + 4$. 
By the union bound, the probability that at least one of  Alice or Bob sends an error flag is at most $2 \delta$, and then
the referee's random output is wrong with probability $1/2$ on any $(x, y)$, so that the probability of error of $\Pi^\prime$ for computing $f$ is at most $\epsilon + \delta$, completing the proof.
\end{proof}

% SSSSSSSSSSSSSSSSSSSSSSSSSSSSSSSSSSSSSSSSSSSSSSSSSS %
\subsection{Link between $IL$ and $CC_{priv}$ (Lemma~\ref{lem:ilvsccpriv})}
% SSSSSSSSSSSSSSSSSSSSSSSSSSSSSSSSSSSSSSSSSSSSSSSSSS %

\begin{lemma}
For any Boolean function $f$, any $\epsilon \geq 0$, any $\delta_1 > 0$, and  any $\delta_2 > 0$ satisfying $\epsilon + \delta_1 + \delta_2 < \frac{1}{2}$, denote $n_A = \log |X|$ and $n_B = \log |Y|$. Then it holds that 
\begin{align*}
IL (f, \epsilon) \geq  \delta_1 \big( CC_{priv} (f, \epsilon + \delta_1 +\delta_2) - g_2 (n_A, n_B, \delta_2) - 4  \big) - 2 g_1 (CC_{priv} (f, \epsilon)),
\end{align*}
with $g_1 (x) = 2 \log (x + 1) + 10$, and
$g_2 (x, y, z) = 2 \log (\frac{2 (x + y)}{z^2 \cdot \log e } +1) + 2  $.
\end{lemma}

\begin{proof}
The lemma follows by combining Lemmata~\ref{lem:icvsccav},~\ref{lem:icvsilworst},~\ref{lem:ccprivvssh} and~\ref{lem:ccshvsav} through the following chain of inequality:
\begin{align*}
IL (f, \epsilon) = & IC(f, \epsilon) \\
		\geq & CC_{av} (f, \epsilon) - 2 g_1 (CC_{priv} (f, \epsilon)) \\
		\geq & \delta_1 \big( CC_{sh} (f, \epsilon + \delta_1) - 4 \big) - 2 g_1 (CC_{priv} (f, \epsilon)) \\
		\geq & \delta_1 \big( CC_{priv} (f, \epsilon + \delta_1 + \delta_2) -  g_2 (n_A, n_B, \delta_2) - 4\big) - 2 g_1 (CC_{priv} (f, \epsilon)).
\end{align*}
\end{proof}

% SSSSSSSSSSSSSSSSSSSSSSSSSSSSSSSSSSSSSSSSSSSSSSSSSS %
\subsection{Lower bound on $CC_{priv} (EQ_n, \epsilon)$ (Lemma~\ref{lem:bbkim})}
% SSSSSSSSSSSSSSSSSSSSSSSSSSSSSSSSSSSSSSSSSSSSSSSSSS %

\begin{lemma}
For any $n \in \mathbb{N}$ and any $\epsilon \in [0, \frac{1}{2})$, the following holds:
\begin{align*}
CC_{priv} (EQ_n, \epsilon) \geq 2 \sqrt{g_3 (\epsilon)} \sqrt{n} - g_3 (\epsilon) - 6,
\end{align*}
with $g_3 (x) = 2 \cdot (1/2 - x)^2 \cdot \log e $.
\end{lemma}

\begin{proof}
We use Hoeffding's inequality: Let $x_1, \cdots, x_t$ be i.i.d random variables in $[0, 1]$, 
and denote the empirical mean $\frac{1}{t}\sum_{i=1}^t x_i = \bar{x}$. 
Then 
\begin{align*}
\Pr [\bar{x} - \mathbb{E} [\bar{x}]  \geq \delta] & \leq \exp_2 (-2 \delta^2 t \cdot \lg e).
\end{align*}

%We denote $n_A = \lceil \log |X| \rceil$, $n_B = \lceil \log |Y| \rceil$, and $t = \lceil \frac{n_A n_B}{\delta^2} \rceil +1$.
Fix $\delta> 0$, and let $c_A = \log |M_A|$,  $c_B = \log |M_B|$. We take $t = \lceil \frac{c_B + 2}{2 \delta^2 \cdot \log e }  \rceil \leq \frac{c_B + 2}{2 \delta^2 \cdot \log e } + 1$.

Given a protocol $\Pi$ in the private coin SMP model, 
we derandomize the protocol on Alice's side, i.e.~we define a protocol $\Pi^\prime$ also in the private coin SMP model, but in which Alice is deterministic.

Given inputs $(x, y)$ and random 
strings $r_A$, $r_B$, and $r_C$ used for one run of 
protocol~$\Pi$, let $m_A (x, r_A)$ be the message sent by Alice on input $x$ and random string $r_A$, and similarly let $m_B (y, r_B)$ denote the message on Bob's side. 
Given messages $m_A$ and $m_B$, let
%\begin{align}
$\Pi (m_A , m_B , r_{C})$
%\end{align}
 be the output of the referee when run 
on messages $m_A$ and $ m_B$ and random string $r_{C}$, and denote
\begin{align}
Q (m_A, m_B) = \Pr_{r_C} [\Pi (m_A, m_B, r_C) = 1]. 
%with the probability taken over the random string $r_C$.
\end{align}

Now, fix $x$, $m_B$ and $t$ random strings $r_A^1, \cdots, r_A^t$ and let $m_A^i (x) = m_A (x, r_A^i)$. Denote 
\begin{align}
P (x, m_B) = \bE_{r_A} [Q (m_A (x, r_A), m_B) ]
\end{align}
%, with the probability taken over the random strings $r_A$ and $r_C$, and 
and 
\begin{align}
\bar{Q} (m_A^i (x), m_B) = \frac{1}{t} \sum_{i=1}^t Q (m_A^i (x), m_B).
\end{align}
We get by 
Hoeffding's inequality that 
\begin{align}
\Pr [\bar{Q} (m_A^i (x), m_B) - P (x, m_B) \geq  \delta] \leq \exp_2 (-2 \delta^2 t \cdot \log e ),
\end{align}
in which the probability is over i.i.d. $r_A^i \sim R_A$ being used to generate each  $m_A^i (x)$. 
With the above choice for $t$, we get that for any input $x$ and message $m_B$, 
\begin{align} 
\Pr [\bar{Q} (m_A^i (x), m_B) -  P (x, m_B) \geq \delta] < \exp_2(- c_B -1).
\end{align} 
We similarly get, using a symmetrical argument, that 
\begin{align}
 \Pr[P (x, m_B)  - \bar{Q} (m_A^i (x), m_B)  \geq \delta] < \exp_2(- c_B -1),
\end{align}
and by the union bound we get that
\begin{align}
 \Pr [| \bar{Q} (m_A^i (x), m_B) -  P (x, m_B) | \geq \delta] < \exp_2(- c_B). 
\end{align}
For any fixed $x$, this holds for any message $m_B$ of Bob, and we get by the union bound 
that there exists a choice of the $r_A^i (x)$, and corresponding $m_A^i (x)$, such that for all $m_B \in M_B$,
\begin{align}
\label{eq:ebar}
| \bar{Q} (m_A^i (x), m_B) -  P (x, m_B) | <   \delta.
\end{align}

Consider now the following protocol $\Pi^\prime$ in which Alice is deterministic.
Given $x$, Alice finds $t$ random strings $\{ r_A^i (x) \}_{i=1}^t$ 
and corresponding messages $m_A^i (x) = m_A (x, r_A^i (x))$ such 
that (\ref{eq:ebar}) holds for all possible messages $m_B$ of Bob. 
She sends as her deterministic message $m_A^\prime (x)$ the concatenation $m_A^1 (x) \cdots m_A^t (x)$ of these messages, which is of length 
$\lceil t c_A \rceil$. Given input $y$ and random string $r_B^\prime = r_B$, distributed according to $R_B$, 
Bob sends message $m_B^\prime (y, r_B^\prime) = m_B(y, r_B)$, as in protocol $\Pi$.
Upon receiving these messages, the referee computes $\bar{Q} (m_A^i (x), m_B (y, r_B))$, 
and outputs $1$ with probability $\bar{Q} (m_A^i (x), m_B (y, r_B))$ and $0$ otherwise. 
Denote $\Pi^\prime (x, y)$ the random variable corresponding to the output of 
protocol $\Pi^\prime$ on inputs $(x, y)$, and similarly for protocol $\Pi$. 
Then 
\begin{align*}
\Pr [\Pi^\prime (x, y) = 1] & = \bE_{r_B} [\bar{Q} (m_A^i (x), m_B (y, r_B))] \\
				& \leq \bE_{r_B} [P (x, m_B (y, r_B))] + \delta \\
				&= \Pr [\Pi (x, y) = 1] + \delta, 
\end{align*}
so that if $x \not= y$, then $\Pi^\prime$ makes errors at most $\epsilon + \delta$. A similar argument shows that $\Pi^\prime$ makes error at most $\epsilon + \delta$ also when $x = y$, and hence on all inputs.

Now, fix $\epsilon$ and $\delta$ such that $\epsilon + \delta < 1/2$. Since Alice's message in $\Pi^\prime$
is deterministic, either $\exp_2 (t c_A) \geq  2^n$, or there exist $x \not= x^\prime$ such that Alice 
sends the same message on both $x$ and $x^\prime$. But then, for any $y$,
\begin{align*}
\Pr [\Pi^\prime (x, y) = 1] & = \bE_{r_B } [\bar{Q} (m_A^i (x), m_B (y, r_B))] \\
			& = \bE_{r_B \sim R_B} [\bar{Q} (m_A^i (x^\prime), m_B (y, r_B))] \\
			& = \Pr [\Pi^\prime (x^\prime, y) = 1]. 
\end{align*}
Considering $y = x$, it must hold that $\Pr [\Pi^\prime (x, x) = 1] > 1/2$ 
and $\Pr [\Pi^\prime (x^\prime, x) = 1] < 1/2$, a contradiction. Hence, it must hold that $t c_A \geq n$. 
To get the desired bound, fix $\delta_1 > 0$, take $\delta = 1/2 - \epsilon - \delta_1$, and let  $\delta_1$ tend to
zero. Optimizing $c_A$ and $c_B$ (e.g., using calculus methods on continuous relaxations of $c_A$ and $c_B$, each within an additive factor of $2$) 
in order to minimize the total communication $c_A + c_B$, we obtain  $c_A  \geq \sqrt{g_3 (\epsilon) n } - 2 $ 
and $c_B \geq \sqrt{g_3 (\epsilon) n } -  g_3 (\epsilon) - 4$, so that $CC_{priv} (EQ_n, \epsilon) \geq 2 \sqrt {g_3 (\epsilon) n} - g_3 (\epsilon) -  6$.
\end{proof}

\bibliographystyle{alpha}
\bibliography{eqinfoleak2}

\end{document}